\documentclass[preprint,11pt]{elsarticle}
\usepackage[suffix=]{epstopdf}
\usepackage{enumitem}
\usepackage{textcomp}

\usepackage{geometry}
\usepackage{mathtools}
\usepackage{color}

\usepackage{amssymb,amsmath}
\usepackage{amsfonts}
\usepackage{graphics}
\usepackage{amsthm}

\newtheorem{remark}{Remark}
\newtheorem{theorem}{Theorem}
\newtheorem{crit}[theorem]{Criterion}

\newtheorem{coro}[theorem]{Corollary}

\begin{document}

\title{Approximating hidden chaotic attractors via parameter switching}
\vspace{5mm}

\author [rm1,rm2]{Marius-F. Danca}
\author [rm3,rm4]{Nikolay Kuznetsov\corref{cor1}}
\author [rm5]{Guanrong Chen}
\cortext[cor1]{Corresponding author}
\address[rm1]{Department of Mathematics and Computer Science, Avram Iancu University of Cluj-Napoca, Romania}
\address[rm2]{Romanian Institute of Science and Tecçhnology, Cluj-Napoca, Romania}
\address[rm3]{Department of Applied Cybernetics, Saint-Petersburg State University, Russia}
\address[rm4]{Department of Mathematical Information Technology,
 University of Jyv\"{a}skyl\"{a}, Finland}
\address[rm5]{Department of Electronic Engineering, City University of Hong Kong}

\begin{abstract}
In this paper, the problem of approximating hidden chaotic attractors of a general class of nonlinear systems is investigated. The Parameter Switching (PS) algorithm is utilized, which switches the control parameter within a given set of values with the initial value problem numerically solved. The PS-generated generated attractor approximates the attractor obtained by averaging the control parameter with the switched values, which represents the hidden chaotic attractor. The hidden chaotic attractors of a generalized Lorenz system and the Rabinovich-Fabrikant system are simulated for illustration.
\end{abstract}


\maketitle

\textbf{In \cite{danca,dan,dan2} it is proved that the attractors of a chaotic system, considered as the unique numerical approximations of the underlying $\omega$-limit sets (see e.g. [31]), after neglecting sufficiently long transients, can be numerically approximated by switching the control parameter in some deterministic or random manner while the underlying initial value problem (IVP) is numerically integrated with the Parameter Switching (PS) algorithm. The attractors, whose basins of attractions are not connected with equilibria are called \emph{hidden attractors}, while the attractors for which the trajectories starting from a point in a neighborhood of an unstable equilibrium are attracted by some attractor, are called \emph{self-excited attractors} \cite {LeonovKV-2011-PLA,LeonovK-2013-IJBC,LeonovKM-2015-EPJST}. In this paper we prove analytically and verified numerically that the PS algorithm can be used to approximate any desired hidden attractors of a class of general systems which model systems like Lorenz, Chen, R\"{o}ssler, etc.}

\section{Introduction}
One main task in the investigation of a dynamical model
is to study the limiting behavior of the system states after the transient processes,
i.e., the problem of localization and analysis of attractors
(limiting sets of system\textquotesingle s states).
Here, one of the challenging problems is to study models with multistability, whose states
can alternate between some mutually exclusive attractors over time \cite{PisarchikF-2014}.
In such models, particularly in the case
of the existence of attractors with very small basins or unidentified attractors,
one can observe sudden switching to unexpected (unpredictable or unknown) attractors,
since such systems  are  sensitive  to noise,  initial
conditions, and system parameters \cite{MenckHMK-2013,DudkowskiJKKLP-2016}.
While trivial attractors (stable equilibria) can be found either analytically or numerically of any dynamical system, the search for nontrivial attractors could be a very challenging task (e.g. the famous Hilbert 16th problem on periodic attractors in two-dimensional polynomial dynamical systems is still far from being solved).
The structures of many classical physical dynamical models
guarantee that
attractors exist because the trajectories
cannot tend to infinity and the oscillations are excited by an unstable
equilibrium. Such attractors are called \emph{self-excited attractors},
which can be easily found by constructing a solution using initial data
from a small neighborhood of the equilibrium, observing how it is attracted thereby visualizing the attractor.
However, there are attractors of another type, called
\emph{hidden attractors} \cite{LeonovKV-2011-PLA,LeonovK-2013-IJBC,LeonovKM-2015-EPJST}, whose basins of attractions
are not connected with equilibria
and, thus, the search and study of such attractors are very challenging \cite{DudkowskiJKKLP-2016,Kuznetsov-2016}.
For example, hidden attractors can be in systems with no equilibria \cite{a} or in
a multistable system with only one stable equilibrium \cite{b}.

Self-excited attractors can be numerically visualized
through a standard computational procedure, in which after the
transient process a trajectory starting from a point in
a neighborhood of an unstable equilibrium is attracted
to an attractor. In contrast, the basin of attraction for a
hidden attractor is not connected with any small neighborhood of any equilibrium
and, thus for the numerical localization of a hidden
attractor it is necessary to develop a special analytical-numerical procedure, in which an initial point is chosen from the basin of attraction. To numerically verify that a chaotic attractor is hidden, one has to check that all trajectories starting in small neighborhoods of unstable equilibria are either attracted by stable attractors or diverging to infinity.

The known autonomous chaotic dynamical systems depending on a single real control parameter $p\in \mathbb{R}$, such as the Lorenz system, R\"{o}ssler system, Chen system, Lotka--Volterra system, Rabinovich--Fabrikant system, Hindmarsh-Rose system, L\"{u} system, classes of minimal networks and many others, are modeled by the following Initial Value Problem (IVP):
\begin{equation}\label{e0}
\dot x(t)=f(x(t))+pAx(t),\quad x(0)=x_0,
\end{equation}
\noindent where $t\in I=[0,T]$, $x_0\in \mathbb{R}^n$, $p\in \mathbb{R}$ the control parameter, $A\in \mathbb{R}^{n\times n}$ a constant matrix, and $f : \mathbb{R}^n \rightarrow\mathbb{R}^n$ a continuous nonlinear function.

For example, for the Lorenz system
\begin{equation}%
\begin{array}
[c]{cl}%
\overset{\cdot}{x}_{1}= & \sigma(x_{2}-x_{1}),\\
\overset{\cdot}{x}_{2}= & x_{1}(\rho-x_{3})-x_{2},\\
\overset{\cdot}{x}_{3}= & x_{1}x_{2}-\beta x_{3},
\end{array}
\label{lorenz}%
\end{equation}
\noindent with $n=3$ and the standard parameter values $a=10$ and $c=8/3$, if one considers $p=\rho$ then system \eqref{e0} has
\[
f(x)=\left(
\begin{array}
[c]{c}%
\sigma(x_{2}-x_{1})\\
-x_{1}x_{3}-x_{2}\\
x_{1}x_{2}-\beta x_{3}%
\end{array}
\right)  ,~~A=\left(
\begin{array}
[c]{ccc}%
0 & 0 & 0\\
1 & 0 & 0\\
0 & 0 & 0
\end{array}
\right).
\]

The PS algorithm approximates numerically any solution of the IVP (\ref{e0}) \cite{danca,dan,dan2}. If one chooses a finite set of values of the underlying control parameter, $\mathcal{P}_N=\{p_1,p_2,...,p_N\}$, $N\geq2$, and then switches $p$ within $\mathcal{P}_N$ for a relatively short period of time, while the underlying IVP is numerically integrated, then the resultant ``switched'' numerical solution will converge to the ``averaged'' solution of the system. Consequently, any attractor of the underlying system, obtained for $p$ being replaced with the average of switched values, can be approximated by the attractor generated from the switching operations.

The PS algorithm was successfully applied to approximating the attractors of continuous-time chaotic systems of integer or fractional order, including the Lorenz system, Chen system, L\"{u} system, R\"{o}ssler system, Hastings-Powell system, Lorka-Volterra system, minimal networks, Hindmarsh-Rose neuronal system, Rikitake system, etc. \cite{danca,dan2,foot2,danxxx,danxy,danux,dan7,dan8,danzx,dada2,dancatt}, and also to discrete nonlinear systems of real variables \cite{danca_dis1,danca_dis2}, or of complex variables (fractals) \cite{danca_frac,danca_frac2}. Moreover, the algorithm can be utilized in experimental applications \cite{danca_prac} and synchronization \cite{dan_nick}.

The PS algorithm is useful e.g. when one intends to obtain an attractor but for some reason the underlying parameter of the attractor cannot be set. Also, the PS algorithm could explain why, in some natural systems, alternations between different dynamics could lead to unexpected behavior.

In this paper, the PS algorithm is used to approximate some hidden chaotic attractors which, as mentioned above, is a challenging task.

The paper is organized as follows: Section 2 presents the PS algorithm, its convergence and its numerical implementation. In Section 3, the PS algorithm is used to approximate hidden attractors in a generalized Lorenz system and the Rabinovich--Fabrikant system, respectively. The short conclusion section ends the investigation.

\section{Parameter Switching algorithm}

\subsection{Description and convergence of the PS algorithm}

Let $\mathcal{P}_N=\{p_1,p_2,...,p_N\}\subset \mathbb{R}$, a set of $N$ values of parameter-$p$, $N\geq2$. Consider the IVP \eqref{e0}, numerically integrated it over $I$ with $p$ switching periodically its values within $\mathcal{P}_N$ for relatively short periods of time. The PS algorithm is associated with the  ``switching'' equation in the following form:
\begin{equation}\label{e1}
\dot x(t)=f(x(t))+p_h(t)Ax(t),\quad x(0)=x_0,
\end{equation}

\noindent with $p_h:I\rightarrow\mathcal{P}_N$ being a $T_p$-periodic piece-wise constant function, depending on a small $h>0$, which switches periodically its values
$p_h(t)=p_h(t+T_p)=p_i\in\mathcal{P}_N$, $i\in\{1,2,...,N\}$, for $t\in I_{i,j}$, and $j=1,2,...$, where $I_{i,j}$ are subintervals of the time interval $I=\bigcup _{j}  \big(\bigcup _{i=1}^N I_{i,j}\big)$ (see the sketch in Fig. \ref{fig1} for the case of $N=3$ and $\mathcal{P}_3=\{p_1,p_2,p_3\}$).

Denote the average of the switched values by $p^*$, which is a constant having the same value for all $t\in I=[0,T]$:

\begin{equation}
p^*=\frac{1}{T_p}\int_t^{t+T_p}p_h(u)du,~~ t\in I,
\end{equation}
{where $p_h(\cdot)$ is one of the parameter values $p_i$, $i\in \{1,2,...,N\}$.

Then, the ``averaged'' equation of (\ref{e0}), obtained for $p$ being replaced with $p^*$, reads
\begin{equation}\label{e2}
\dot {\bar x}(t)=f(\bar x(t))+p^*A\bar x(t),\quad t\in I=[0,T],\quad \bar x(0)=\bar x_0.
\end{equation}

It can be proved that switching $p$ within $\mathcal{P}_N$ in \eqref{e1} while the switching equation \eqref{e1} is integrated, the obtained solution approximates the solution of the averaged equation \eqref{e2}. Before proceeding, the following assumption is needed.

\noindent\textbf{Assumption H1}. In the IVP (\ref{e0}), $f$ satisfies the Lipschitz condition

\begin{equation}\label{lip}
|f(y_1)-f(y_2)|\le L|y_1-y_2|,~~~\forall y_{1,2}\in \mathbb{R}^n,
\end{equation}

\noindent for some $L>0$.

Denote $p_h(t):=P(t/h)$ and let $\|\cdot\|_0$ be the maximum norm on $C([0,T],\mathbb{R}^n)$,  i.e., $\|\bar x\|_0:=\max_{t\in[0,T]}|\bar x(t)|$. Then, under Assumption H1, on $[0,T]$, the following theorem holds\footnote{The convergence for the infinite interval $[0,\infty)$ requires some dissipativity or stability of the system \cite{kipi}.}.

\begin{theorem}\label{tt1}
Under Assumption H1
\begin{equation}\label{est1}
\begin{gathered}
|x(t)-\bar x(t)|\le ( |x_0-\bar x_0|+h\|A\|\|\bar x\|_0K)
 e^{(L+\|P\|_0\|A\|)T},
\end{gathered}
\end{equation}
for all $t\in [0,T]$, where
$$
K:=\max_{t\in[0,T_p]}\left|\int_0^{t}(P(s)-p^*)ds\right|.
$$
\end{theorem}

\begin{proof}

From \eqref{e1} and \eqref{e2}, one has
$$
\begin{gathered}
|x(t)-\bar x(t)|\le |x_0-\bar x_0|+ L\int_0^t|x(s)-\bar x(s)|ds
+\left|\int_0^t(p_h(s)-p^*)ds\right|\|A\|\|\bar x\|_0\\
+\|P\|_0\|A\|\int_0^t|x(s)-\bar x(s)|ds
\\=|x_0-\bar x_0|+\|A\|\|\bar x\|_0\left|\int_0^t(p_h(s)-p^*)ds\right|
+(L+\|P\|_0\|A\|)\int_0^t|x(s)-\bar x(s)|ds.
\end{gathered}
$$
Because $p_h$ is $T_p$-periodic, one has
\begin{equation*}
\max_{t\in[0,T]}|p_h(t)|\le \max_{t\in[0,T_p]}|P(t)|=\|P\|_0.
\end{equation*}
On the other hand, one has
$$
\int_0^t(p_h(s)-p^*)ds=h\int_0^{t/h}(P(s)-p^*)ds.
$$
Since $\int_0^{t}(P(s)-q^*)ds$ is $T_p$-periodic, one has
$$
\max_{t\in[0,T]}\left|\int_0^{t/h}(P(s)-p^*)ds\right|\le K.
$$
Hence,
$$
\begin{gathered}
|x(t)-\bar x(t)|\leq |x_0-\bar x_0|+h\|A\|\|\bar x\|_0K
+(L+\|P\|_0\|A\|)\int_0^t|x(s)-\bar x(s)|ds.
\end{gathered}
$$
By the Gronwall inequality \cite{HH}, one obtains \eqref{est1}.
\end{proof}

\begin{remark}

\begin{itemize}
\item [i)] The above proof is more general than the proof presented in \cite{dan2}, where the convergence is obtained via the averaging method \cite{ver} and the initial conditions of \eqref{e1} and \eqref{e2} are equal. In \cite{dan}, the proof is made numerically on the basis of the global error of Runge-Kutta. In \cite{kipi}, beside the convergence of the PS algorithm, numerical approximation estimation and Lyapunov method are presented, and moreover the PS convergence for any utilized Runge-Kutta method is proved.
\item[ii)]The periodicity assumption on $p$ in \eqref{e1} is too strong. Actually, the convergence proof in \cite{dan} and the numerical experiments in \cite{dancatt,foot2,danxxx} (or experimental applications in \cite{danca_prac}) show that the PS algorithm can be implemented in some random way as well. For instance, once $\mathcal{P}_N$ is set, the order in which $p$ switches its values, $p=p_i\in\mathcal{P}_N$, can be random. Therefore, one may assume that random or periodic switches of parameters in natural systems have a real meaning, such as in ecological systems or circuitry. Also, random parameter switches in some systems explain why chaotic (hidden) attractors could appear unexpectedly.
\item[iii)]  If the averaged system has a hyperbolic invariant compact set, then the switching equation (3) has also a near hyperbolic invariant compact set.
    \end{itemize}
\end{remark}

A global attractor is a compact and invariant set composed of all bounded global trajectories and contains all the dynamics evolving from all possible initial conditions. In other words, it contains all solutions, including stationary solutions, periodic solutions, as well as chaotic attractors, relevant to the asymptotic behaviors of the system. On the contrary, a local attractor is a compact invariant set, which attracts its neighboring trajectories. A global attractor is hence composed of the set of all local attractors, where each local attractor only attracts trajectories from a subset of initial conditions, specified by its basin of attraction. In some cases, a unique local attractor may also be the global one. When  a  global  attractor  is  composed of several local attractors, the initial conditions are essential for the numerical approximations of these attractors, respectively. Therefore, the following assumption is made.

\noindent \textbf{Assumption H2} Suppose that $x_0$ and $\bar{x}_0$ belongs to the same attraction basin of solutions to the IVP \eqref{e1}.

The $\omega$-limit set of a trajectory through $x\in \mathbb{R}^n$ is given as $\omega(x)=\bigcap_{s\geq0}\overline{\bigcup_{t\geq s}\Phi(t,x)}$, where $\Phi(t,x)$ is the flow of the system.

As common in numerical investigations of nonlinear systems, for every $p$ and $x_0$, by an \emph{attractor} one considers the unique numerical approximation of the underlying $\omega$-limit sets (see e.g. \cite{foi}), neglecting sufficiently long transients.

By Theorem \ref{tt1}, which characterizes the PS algorithm, the following result can be obtained.

\begin{coro}\label{teo}
Every attractor of the system modeled by the IVP \eqref{e1} can be numerically approximated using the PS algorithm.
\end{coro}

In other words, using the PS algorithm, the attractor $A^*$ (\emph{switched attractor}) obtained from equation \eqref{e1} by switching $p$ within $\mathcal{P}_N$, will approximate numerically the attractor denoted $A_{p^*}$ (\emph{averaged attractor}) obtained from \eqref{e2}.

\subsection{Implementation of the PS algorithm}

To implement numerically the PS algorithm, let $h$ be the step-size of the utilized explicit numerical method for integrating the corresponding IVP (such as the standard Runge-Kutta method, used in this paper).

Symbolically, for a given $h>0$, the PS algorithm can be denoted as

\begin{equation}\label{s0}
PS:=[m_1p_1,m_2p_2,...,m_Np_N],
\end{equation}

\noindent where $m_i\in \mathbb{N}^*$, $i=1,2,...,N$, are some positive integers, called ``weights'' of the $p$ values. Then, $p^*$ can be expressed as
\begin{equation}\label{p}
p^*:=\frac{\sum_{i=1}^Nm_ip_i}{\sum_{i=1}^Nm_i}.
\end{equation}

\noindent The scheme \eqref{s0} reads as follows: while the IVP \eqref{e1} is numerically integrated, for the first $m_1$ integration steps, $p=p_1$; for the next $m_2$ steps, $p=p_2$; and so on, till the last $m_N$ step, when $p=p_N$. So, the first set of subintervals $I_{i,1}$, for $i=1,...,N$, is covered. Next, the algorithm repeats on the next set of subintervals $I_{i,2}$, and so on, until the entire time interval is covered. Time subintervals $I_{i,j}$ have lengths $m_ih$, for $i=1,2,...,N$ and $j=1,2,...,$ and the switching period is $T_p=\sum_{i=1}^Nm_ih$.

\begin{remark}
\begin{itemize}

The main characteristic of the PS algorithm relies on the linear dependence on $p$ of the right-hand side of system (\ref{e0}) and on the convexity of the relation \eqref{p}. By denoting $\alpha_j=m_j/\sum_{i=1}^N m_i$, $j=1,2,...,N$, relation \eqref{p} becomes $p^*=\sum_{i=1}^N\alpha_i p_i$ with $\sum_{i=1}^N\alpha_i=1$. Thus, for any set $\mathcal{P}_N$, $N>1$, and any weights $m_i$, $i=1,2,...,N$, $p^*$ is always inside the interval $(p_{min},p_{max})$, with $p_{min} \equiv min\{\mathcal{P}_N\}$ and $p_{max} \equiv max\{\mathcal{P}_N\}$. Therefore, to approximate some attractor $A_{p^*}$ using the PS algorithm, the set $\mathcal{P}_N$ has to be chosen such as

\begin{equation}\label{cond}
p^*\in(p_{min},p_{max}).
\end{equation}

\end{itemize}
\end{remark}

Consider a dynamical system modeled by the IVP \eqref{e1}, with the set of its attractors $\mathcal{A}$ and some set of admissible parameter values $\mathcal{P}_N$, $N\geq2$. Based on \eqref{cond} and on the convexity of the relation \eqref{p}, beside the fact that every attractor can be approximated with the PS algorithm (Corollary \ref{teo}), the following important result can be proved.

\begin{coro}\label{cor}
Given a set $\mathcal{P}_N$, with $N\geq2$ and weights $m_i$, $i=1,2,\cdots, N$, the attractor $A^*$ obtained with the PS algorithm belongs to $\mathcal{A}$.
\end{coro}

To visualize the results, i.e., to underline the match between the averaged attractor, $A_{p^*}$, which is to be approximated and the approximating attractor, $A^*$, a computer-graphic criterion is now introduced.

\begin{crit}\label{crt} Two attractors are considered to be almost-identical if
\begin{enumerate}[nolistsep]
\item [a.]\label{crt1}their geometrical forms in the phase space (almost) coincide;
\item [b.]\label{crt2}the orientation of the motion is preserved.
\end{enumerate}
\end{crit}

The above criterion is a suitable modification and adaptation of the known concept of \emph{topological equivalence} (see e.g. \cite{hal}), for practical use rather than for theoretical rigor.

Criterion \ref{crt}.\emph{b} is easy to implement computationally (e.g. with Matlab comet3 function) and has been verified for all examples studied later in this paper.

To apply Criterion \ref{crt}.\emph{a} (the match between the two attractors), $A_{p^*}$ (blue or green plot) and $A^*$ (red plot) are overplotted in the phase space and also for their Poincar\'{e} sections. Visually, the histograms reveal the match between attractors.

Also, the match between the two attractors can be verified by calculating the Hausdorff distance $D_H(A^*,A_{p^*})$ between them. The Hausdorff distance between two sets $A$ and $B$ in the metric space $\mathbb{R}^3$, $D_H(A,B)$, is given by \cite[p.~114]{falc}
\begin{equation}\label{haus}
D_H(A,B)=\max\bigg\{\adjustlimits\sup_{x\in A}\inf_{y\in B}d(x,y),\adjustlimits\sup_{y\in B}\inf_{x\in B}d(x,y)\bigg\},
\end{equation}
where $d(a,b)$ is the Euclidean distance between two points $a=(x_1,x_2,x_3)$ and $b=(y_1,y_2,y_3)$ in $\mathbb{R}^3$.
Since the two numerically generated attractors $A^*$ and $A_{p^*}$ are curves with the same number of $M$ ordered pairs of coordinates $A^*=\{a_1,a_2,...,a_M\}$ and $A_{p^*}=\{b_1,b_2,...,b_M\}$, the distance between a point $a_i\in A^*$ to the set $A_{p^*}$ is given by

$$
d(a_i,A_{p^*})=\min_j \|b_j-a_i\|,
$$
for $i,j=1,2,...,M$. Therefore, the Hausdorff distance \eqref{haus} can be calculated numerically by

$$
D_H(A^*,A_{p^*})=\max\bigg\{\max_i \{d(a_i,A_{p^*})\},\max_j \{d(b_j,A^*)\}\bigg\}.
$$

A study of numerical limitations of the PS algorithm is presented in \cite{dada2} and \cite{dan}.

\section{Hidden chaotic attractors approximated with the PS algorithm}

In this section, under Assumptions H1 and H2, hidden chaotic attractors of a generalized Lorenz system and the Rabinovitch-Fabrikant system are computed with the PS algorithm. The approximated hidden chaotic attractor is the averaged attractor $A_{p^*}$. To approximate $A_{p^*}$, one chooses a set of $N$ parameter values for the switching process, $\mathcal{P}_N$, such that condition \eqref{cond} holds and, with weights $m_i$, $i=1,2,...,N$, relation \eqref{p} is verified. Then, scheme \eqref{s0} is applied to obtain the switched attractor $A^*$ which approximates the desired hidden attractor $A_{p^*}$.

The numerical and simulations results of the PS algorithm were realized with the standard RK algorithm, which allows to implement easily the switches imposed by the algorithm for every $m_i$ steps, $i=1,2,...,N$. The integration step-size was taken as $h=0.0002-0.001$, the histograms for the $x_1$ component use 512 bars, and the integration time interval is $I=[0,300]$\footnote{Longer integration time intervals could reveal possible long-time transient chaotic behavior (see e.g. \cite{greb}). However, too large the time intervals could lead to inaccurate numerical solutions (see e.g. \cite{lori}).}.

For the case of stable cycles, $D_H(A^*,A_{p^*})$ is of order $10^{-3}-10^{-2}$. In the case of chaotic attractors, $D_H(A^*,A_{p^*})$ is larger, e.g., for $I=[0,300]$, $D_H(A^*,A_{p^*})$ is of order $10^{-1}$, while for $I=[0,500]$, $D_H(A^*,A_{p^*})$ diminishes at $10^{-2}$. 

In the case of chaotic attractors, in order to reduce numerical errors, Assumption H2 is strengthened by using identical initial conditions $x_0$ and $\bar{x}_0$.

\subsection{Generalized Lorenz system}

Consider the following generalized Lorenz system \cite{n1,n2}, which was obtained from a Rabinovich system \cite{n3,n4}:
\begin{equation}\label{lorand}
\begin{array}
[c]{cl}%
\overset{\cdot}{x}_{1}= & ap(x_1-x_2)-ax_2x_3,\\
\overset{\cdot}{x}_{2}= & px_1-x_2-x_1x_3,\\
\overset{\cdot}{x}_{3}= & -x_3+x_1x_2,
\end{array}
\end{equation}
\noindent with $a<0$.

The system reads as \eqref{e0}, with
\[
f(x)=\left(
\begin{array}
[c]{c}%
-ax_{2}x_{3}\\
-x_{2}-x_{1}x_{3}\\
-x_{3}+x_{1}x_2%
\end{array}
\right)  ,~~A=\left(
\begin{array}
[c]{ccc}%
a & -a & 0\\
1 & 0 & 0\\
0 & 0 & 0
\end{array}
\right).
\]

With $a=-0.5$, the bifurcation diagram, for $p\in[0,40]$ is presented in Fig. \ref{fig2}.

As shown in \cite{n3,n4}, the system reveals hidden chaotic attractors. In \cite{dan00} hidden attractors of the fractional-order case are studied.
In this paper, the hidden attractor,
$H_1$ (see Fig. \ref{fig3}), was obtained for $p^*=7$ (Similar hidden chaotic behavior was found for
some $p^* > 7$) with equilibria $$X_0^*(0,0,0),~~~ X_{1,2}^*(\pm3.533,\pm1.834,6.481).$$

Since the eigenvalues of $X_0^*$ are $(-7.355,-1,2.855)$, the equilibrium $X_0^*$ is a saddle. $X_{1,2}^*$ are stable (focus node) equilibria since their eigenvalues are $(-5.497,-0.002 + 3.900i,-0.002 - 3.900i)$.
The zoomed vicinity $V_{X_0^*}$ of the unstable equilibrium $X_0^*$ reveals the fact that all trajectories started from $V_{X_0^*}$ are attracted either by the stable equilibrium $X_1^*$ (red trajectories) or by the stable equilibrium $X_2^*$ (blue trajectories). Therefore, the chaotic attractor is a hidden attractor.

\newcounter{saveenum}
\begin{enumerate}[wide, label=\emph{Example \arabic*.},labelindent=0pt]
\item  A stable cycle corresponding to $p^*=25.5$ is obtained, situated in a relative large periodic window (Fig. \ref{fig2}). The attractor, considered as the averaged attractor $A_{p^*}$, can be obtained using e.g. scheme \eqref{s0}, with e.g. $$[1p_1,1p_2,1p_3,1p_4,2p_5],~~\text{and}~~\mathcal{P}_5=\{6.5,22.2,28,31.9,32.2\},$$
    where $m_1=m_2=...=m_4=1$ and $m_5=2$. This gives, via \eqref{p}, $p^*=(1\times p_1+1\times p_2+1\times p_3+1\times p_4+2\times p_5)/(1+1+1+1+2)=25.5$. The switching period is $T_p=\sum_{i=1}^5 m_ih=6h$.

     To underline the perfect match between the obtained attractor $A^*$ (red plot) and the (stable cycle) averaged attractor (blue plot) $A_{p^*}$, in Fig. \ref{fig4} the two attractors are overplotted in the phase space (Fig. \ref{fig4} a) and in Poincar\'{e} sections with $x_3=28$ (Fig. \ref{fig4} b), respectively. Histograms relating to the component $x_1$, for both $A^*$ and $A_{p^*}$, are plotted in Fig. \ref{fig4} (c) and (d), respectively, where transients have been removed.
\item  \label{xy}Other sets of $\mathcal{P}_N$ with appropriate weights can be used to obtain the same stable cycle, using for example the scheme
    $$[1p_1,1p_2],~~\text{with}~~\mathcal{P}_2=\{21,30\},$$
    where $m_1=m_2=1$, which gives $p^*=(p_1+p_2)/2=25.5$. In this case, $A_{p^*}$ is obtained by alternating every integration step, $p$, within $\mathcal{P}_2=\{21,30\}$.
\item  \label{exx}Not only stable cycles can be approximated by the PS algorithm, but also chaotic (self-excited) attractors, for example the one corresponding to $p=34.2$ (see Fig. \ref{fig2}), can be approximated by using e.g.
    $$[2p_1,3p_2],~~\text{with}~~ \mathcal{P}_2=\{25.5,40\},$$
    which gives $p^*=34.2$ (Figs. \ref{fig5} (a)-(d); where the Poincar\'{e} sections are obtained for $x_3=38$). Because of the infinite integration time needed to generate a chaotic attractor, between the two attractors there are some small differences as can be seen from the histograms (Figs. \ref{fig5} (c), (d)).


\item To generate the hidden attractor $H_1$, one needs to find a scheme \eqref{s0} which gives $p^*=7$. One of the possible choices is
$$[1p_1,1p_2],~~\text{with}~~\mathcal{P}_2=\{5,9\}.$$
The match between the switched attractor $A^*$ and the averaged attractor $A_{p^*}$ is underlined by the phase overplot of the averaged (hidden) attractor (green plot) and switched attractor (red plot) in Fig. \ref{fig6} (a), overplots of Poincar\'{e} sections with $x_3=6$ (Fig. \ref{fig6} (b)), and the histograms (Figs. \ref{fig6} (c), (d)).
    \setcounter{saveenum}{\value{enumi}}
\end{enumerate}

\vspace{2mm}

\subsection {Rabinovich-Fabrikant system}

The Rabinovich-Fabrikant (RF) system is modeled by the following IVP \cite{dancau}:

\begin{equation}
\label{rf}
\begin{array}{l}
\overset{.}{x}_{1}=x_{2}\left( x_{3}-1+x_{1}^{2}\right) +ax_{1}, \\
\overset{.}{x}_{2}=x_{1}\left( 3x_{3}+1-x_{1}^{2}\right) +ax_{2}, \\
\overset{.}{x}_{3}=-2x_{3}\left( p+x_{1}x_{2}\right),
\end{array}%
\end{equation}

\noindent with $a=0.1$.

 The RF system has extremely rich dynamics, presenting coexisting attractors, self-excited attractors, hidden-attractors and virtual saddles-like equilibria \cite{danx,dancau}.

In \eqref{rf}, one has
\[
f(x)=\left(
\begin{array}
[c]{c}%
x_{2}\left( x_{3}-1+x_{1}^{2}\right)\\
x_{1}\left( 3x_{3}+1-x_{1}^{2}\right)\\
-2x_{3}\left( p+x_{1}x_{2}\right)
\end{array}
\right)  ,~~A=\left(
\begin{array}
[c]{ccc}%
a & a & 0\\
0 & 0 & 0\\
0 & 0 & 0
\end{array}
\right).
\]
One of the two hidden attractors, $H_2$ (Fig. \ref{fig7}), corresponds to $p=0.2876$. For this value of $p$, the equilibria are
$$X_0^*(0,0,0),~~~ X_{1,2}^*(\mp1.1600,\pm0.2479,0.1223),~~~X_{3,4}^*(\mp0.0850,\pm3.3827,0.9953).$$
The equilibrium $X_0^*$ is unstable (saddle) since its eigenvalues are $(-0.5752,0.1-i,0.1+i)$. Equilibria $X_{3,4}$ are also unstable (saddle) with eigenvalues $(0.1869,-0.281 + 5.397i,-0.281 - 5.397i)$, while equilibria $X_{1,2}^*$ are stable (focus nodes), with eigenvalues $(-0.2561,-0.060 - 1.473i, -0.060 - 1.473i)$.
As can bee seen, trajectories from a small vicinity of the unstable equilibrium $X_0^*$ or $X_{3,4}$ are attracted either by infinity or by the stable equilibria $X_{1,2}^*$. Compared with the hidden chaotic attractor $H_1$ of the Lorenz system \eqref{lorand}, due to the presence of complex eigenvalues, in the case of the hidden chaotic attractor $H_2$ here, trajectories starting from $X_0^*$ (grey and black) and also from $X_{3,4}^*$ (light brown and blue, for the case of $X_4^*$) exit the vicinities by spiralling routes (see detail in Fig. \ref{fig3} and Fig. \ref{fig7} (b) and (c)).

To easily choose the set $\mathcal{P}_N$, the bifurcation diagram for $p\in[0.24,0.295]$ may be utilized (Fig. \ref{fig8}).

\begin{enumerate}[wide, label=\emph{Example \arabic*.},labelindent=0pt]
\setcounter{enumi}{\value{saveenum}}
\item To generate $H_2$, one can use e.g. the scheme
$$[1p_1,2p_2,2p_3],~~\text{with}~~P_3=\{0.28,0.289, 0.29\},$$
for which, by \eqref{p}, $p^*=0.2876$. This generates the hidden attractor $H_2$. The match between the obtained switched attractor $A^*$ (red plot) and the approximated hidden attractor $H_2$ (green plot) can be observed from the match in the phase plot (Fig. \ref{fig9} (a)), Poincar\'{e} sections with $x_3=0.35$ (Fig. \ref{fig9} (b) and histograms (Fig. \ref{fig9} (c) and (d)).

\item Similarly, the other hidden attractor $H_3$ (Fig. \ref{fig8})\footnote{Details about this particular attractor can be found in \cite{dancau}.}, which corresponds to $p=0.2715$, can be be obtained with the PS algorithm, by alternating e.g. the values of the set $\mathcal{P}_2=\{0.265, 0.278\}$, using the scheme
    $$[1p_1,1p_2].$$
     In this case, the Poincar\'{e} section is set at $x_3=0.3$. Again, there exists a perfect match between the obtained attractor $A^*$ and the hidden attractor $H_2$ (see the phase overplots, overploted Poincar\'{e} sections and histograms in Figs. \ref{fig10} (a)-(d), respectively).

\end{enumerate}

\section*{Conclusion}

In this paper, it has been proved and verified numerically that the hidden chaotic attractors of dynamical systems modeled by a general initial value problem can be approximated by switching the control parameter, while the problem is integrated. The approximations is verified with numerical tools by means of phase portraits, histograms, Poincar\'{e} sections and Hausdorff distance. In order to facilitate the choice of the switching parameters, the bifurcation diagrams are also utilized. The algorithm has been applied successfully to a generalized Lorenz system and the Rabinovich-Fabrikant system.

\bigskip
\noindent
{\bf Acknowledgments} MFD and NK thank the Russian Scientific Foundation (project 14-21-00041). GC appreciates the GRF Grant CityU 11234916 from the Hong Kong research Grants Council.


\newpage

\begin{figure}
\begin{center}
  \includegraphics[width=0.9\linewidth] {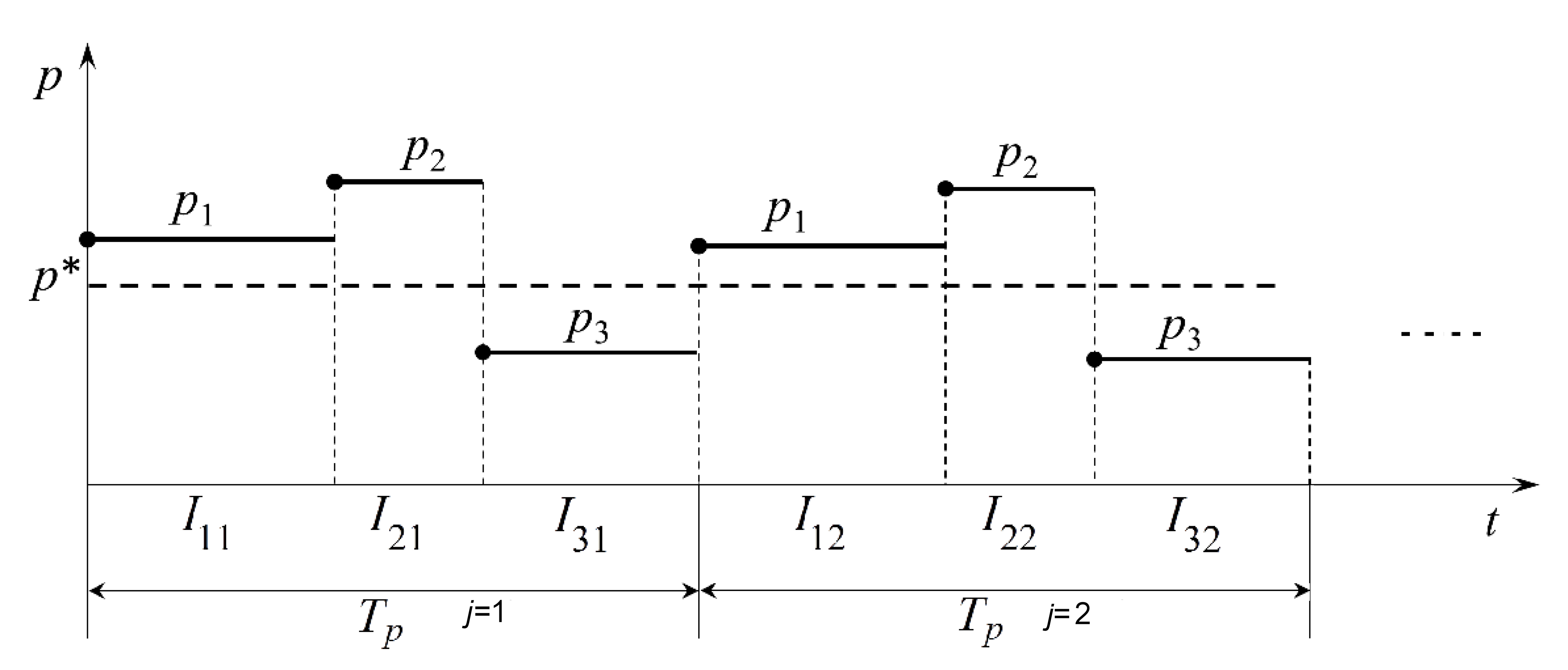}
\caption{Piece-wise constant function $p$ for $N=3$ (sketch).}
\label{fig1}
\end{center}
\end{figure}

\begin{figure}
\begin{center}
\includegraphics[width=0.9\linewidth] {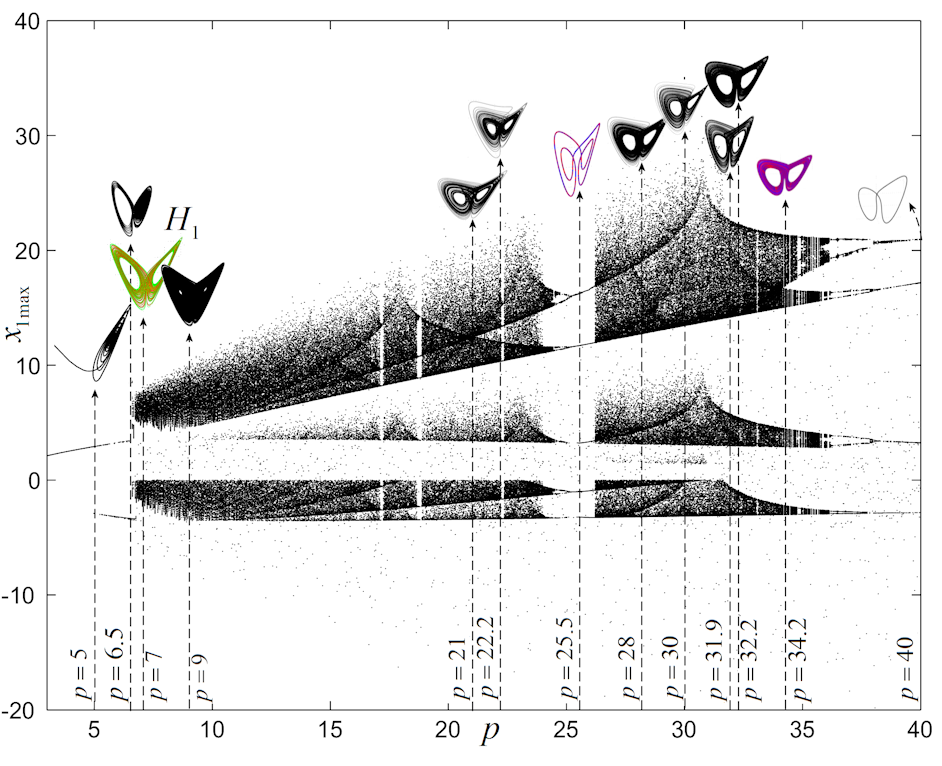}
\caption{Bifurcation diagram of the generalized Lorenz system \eqref{lorand}.}
\label{fig2}
\end{center}
\end{figure}

\begin{figure}
\begin{center}
\includegraphics[width=0.8\linewidth] {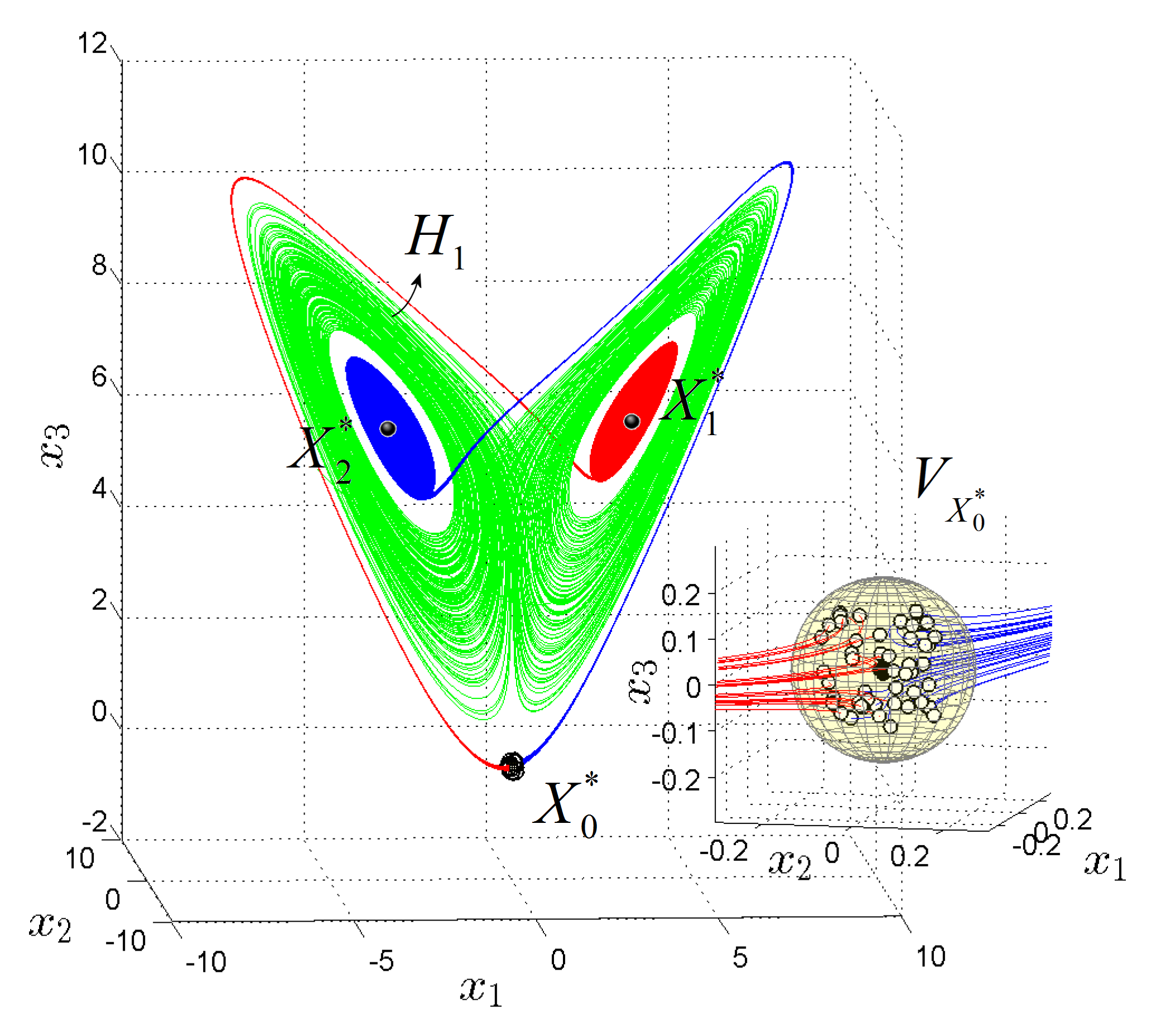}
\caption{Hidden chaotic attractor $H_1$ (green) of the generalized Lorenz system \eqref{lorand}. Trajectories starting from the vicinity $V_{X_0^*}$ of the unstable equilibrium $X_0^*$ are attracted either to the stable equilibrium $X_1^*$ (red plot) or to the stable equilibrium $X_2^*$ (blue plot).}
\label{fig3}
\end{center}
\end{figure}

\begin{figure}
\begin{center}
\includegraphics[width=0.8\linewidth] {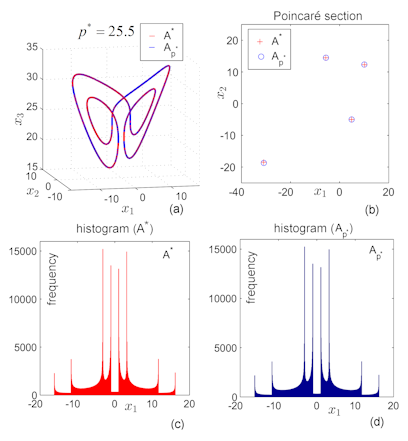}
\caption{Stable cycle of the generalized Lorenz system \eqref{lorand} corresponding to $p^*=25.5$, obtained using the PS algorithm, with the scheme $[1p_1,1p_2,1p_3,1p_4,2p_5]$, $\mathcal{P}_5=\{6.5,22.2,28,31.9,32.2\}$. (a) Overplots of generated attractor $A^*$ (red plot) and averaged attractor $A_{p^*}$ (blue plot). (b) Overplots of Poincar\'{e} sections with plane $x_3=28$, corresponding to $A^*$ and $A_{p^*}$. (c) Histogram with 512 bars of the first component $x_1$ of $A^*$ (red plot). (d) Histogram with 512 bars of the first component $x_1$ of $A_{p^*}$ (red plot).}
\label{fig4}
\end{center}
\end{figure}

\begin{figure}
\begin{center}
\includegraphics[width=0.8\linewidth] {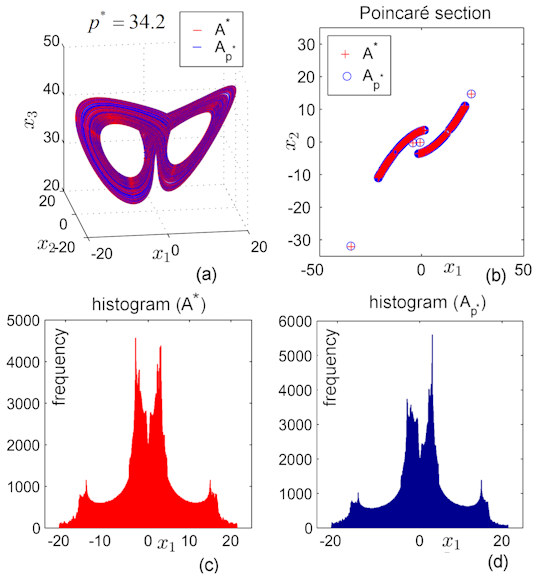}
\caption{Chaotic attractor of the generalized Lorenz system \eqref{lorand} corresponding to $p^*=34$, obtained using the PS algorithm, with scheme $[1p_1,1p_2]$, $\mathcal{P}_2=\{21,30\}$. (a) Overplots of generated stable cycle $A^*$ (red plot) and averaged stable cycle $A_{p^*}$ (blue plot). (b) Overplots of Poincar\'{e} sections with plane $x_3=38$, corresponding to $A^*$ and $A_{p^*}$. (c) Histogram with 512 bars of the first component $x_1$ of $A^*$ (red plot). (d) Histogram with 512 bars of the first component $x_1$ of $A_{p^*}$ (blue plot).}
\label{fig5}
\end{center}
\end{figure}

\begin{figure}
\begin{center}
\includegraphics[width=0.85\linewidth] {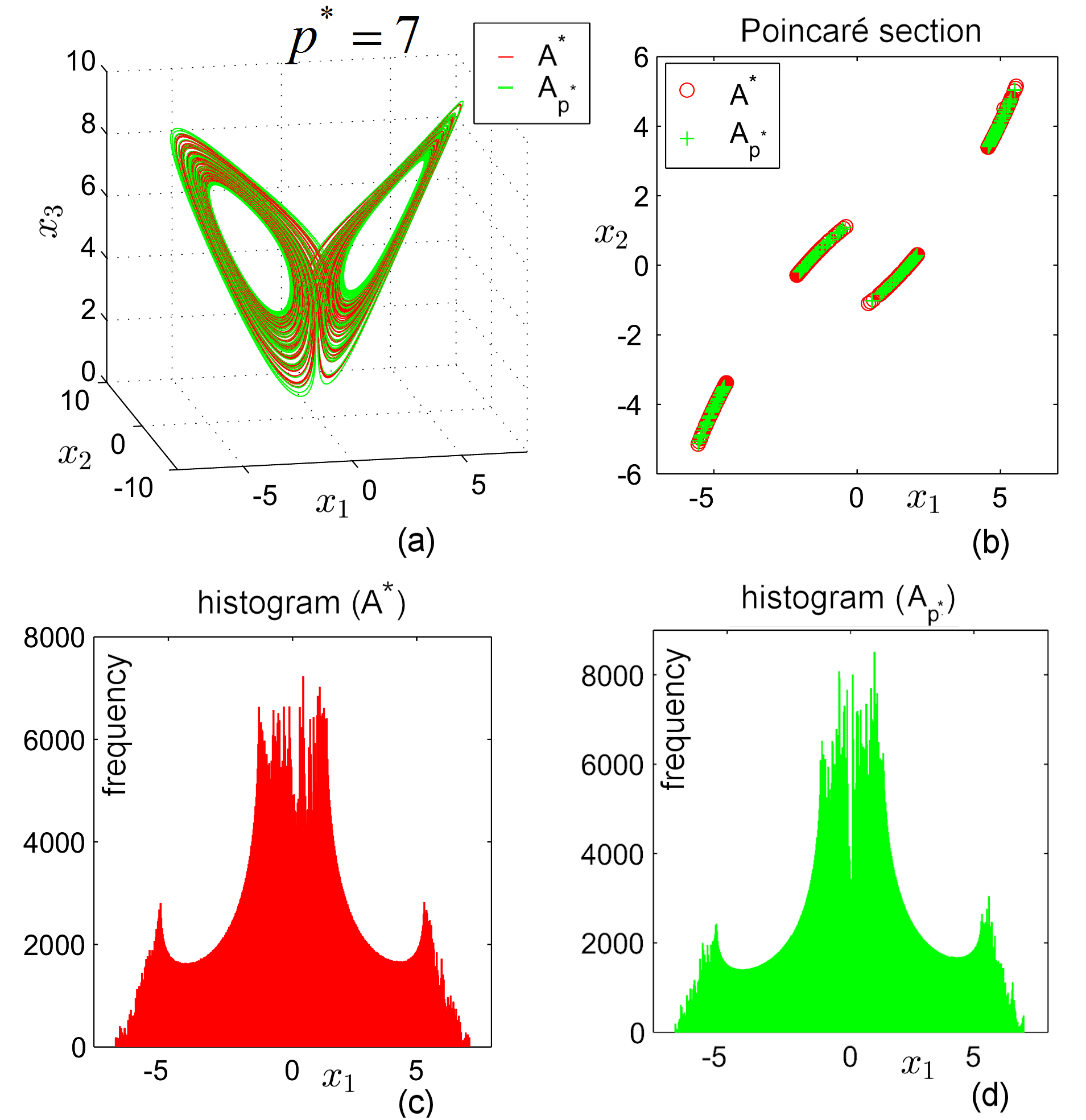}
\caption{Hidden chaotic attractor $H_1$ of the generalized Lorenz system \eqref{lorand} corresponding to $p^*=7$, obtained using the PS algorithm, with scheme $[1p_1,1p_2]$, $\mathcal{P}_2=\{21,30\}$. (a) Overplots of the generated attractor $A^*$ (red plot) and averaged attractor $A_{p^*}$ (green plot). (b) Overplots of Poincar\'{e} sections with plane $x_3=6$, corresponding to $A^*$ and $A_{p^*}$. (c) Histogram with 512 bars of the first component $x_1$ of $A^*$ (red plot). (d) Histogram with 512 bars of the first component $x_1$ of $A_{p^*}$ (green plot).}
\label{fig6}
\end{center}
\end{figure}

\begin{figure}
\begin{center}
\includegraphics[width=0.75\linewidth] {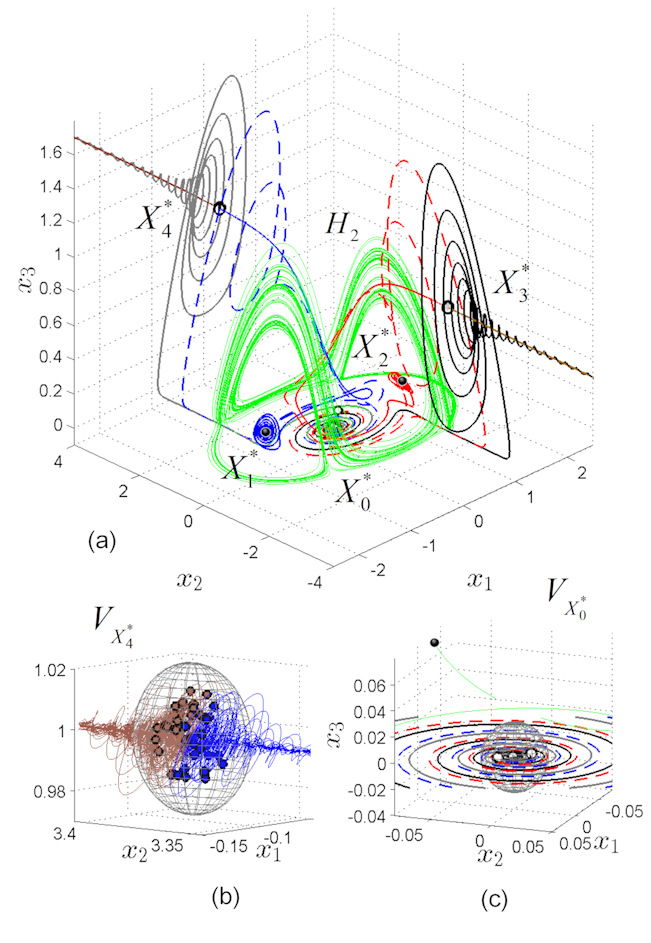}
\caption{Hidden chaotic attractor $H_2$ (green plot) of the RF system \eqref{rf} for $p^*=0.2876$. (a) Phase portrait. (b) Zoomed vicinity $V_{X_4^*}$ of the unstable equilibrium $X_4^*$. (c) Zoomed vicinity $V_{X_0^*}$ of the unstable equilibrium $X_0^*$. Trajectories starting from the unstable points $X_0^*$ and $X_{3,4}^*$ are either attracted to the stable equilibria $X_{1,2}$ (dotted blue and red respectively) or to infinity (light and dark brown, and grey and black).}
\label{fig7}
\end{center}
\end{figure}

\begin{figure}
\begin{center}
\includegraphics[width=0.9\linewidth] {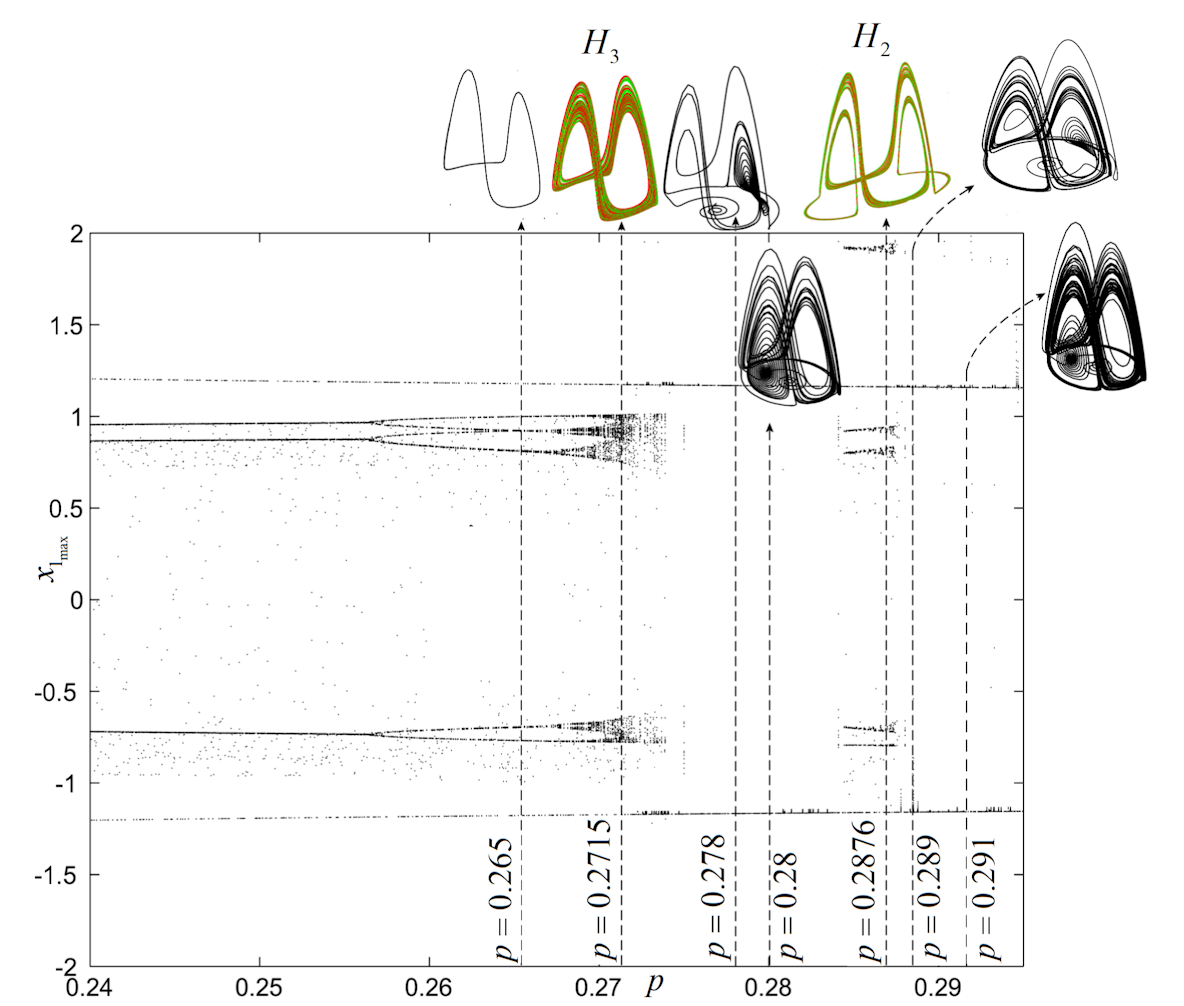}
\caption{Bifurcation diagram of the RF system \eqref{rf} for $p\in[0.24,0.3]$.}
\label{fig8}
\end{center}
\end{figure}

\begin{figure}
\begin{center}
\includegraphics[width=0.9\linewidth] {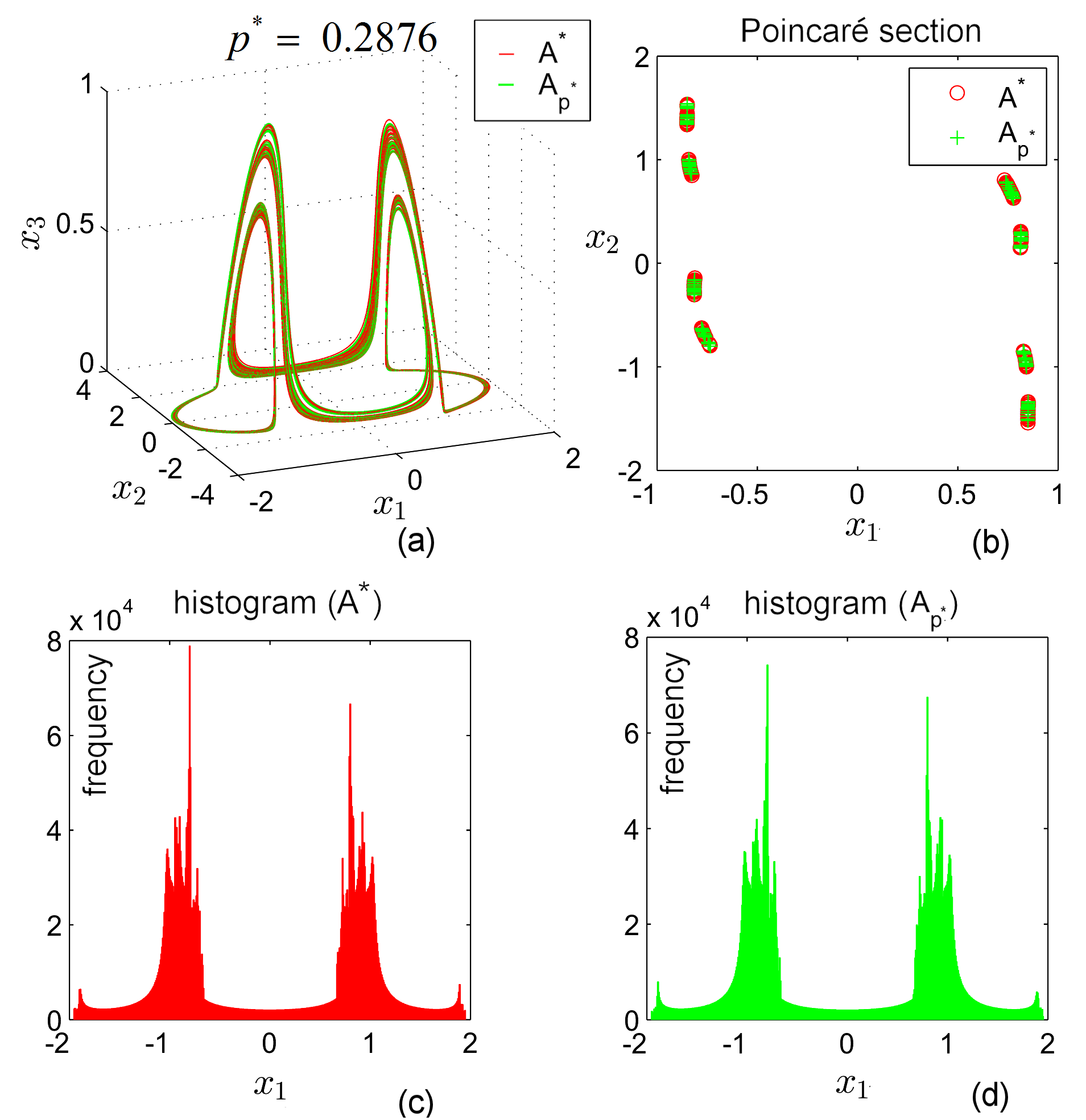}
\caption{Hidden chaotic attractor $H_2$ of the RF system \eqref{rf} corresponding to $p^*=0.2876$, obtained  using the PS algorithm, with scheme $[1p_1,2p_2,2p_3]$, $\mathcal{P}_3=\{0.28,0.289,0.29\}$. (a) Overplots of generated attractor $A^*$ (red plot) and averaged attractor $A_{p^*}$ (green plot). (b) Overplots of Poincar\'{e} sections with plane $x_3=0.35$, corresponding to $A^*$ and $A_{p^*}$. (c) Histogram with 512 bars of the first component $x_1$ of $A^*$ (red plot). (d) Histogram with 512 bars of the first component $x_1$ of $A_{p^*}$ (green plot).}
\label{fig9}
\end{center}
\end{figure}

\begin{figure}
\begin{center}
\includegraphics[width=0.9\linewidth] {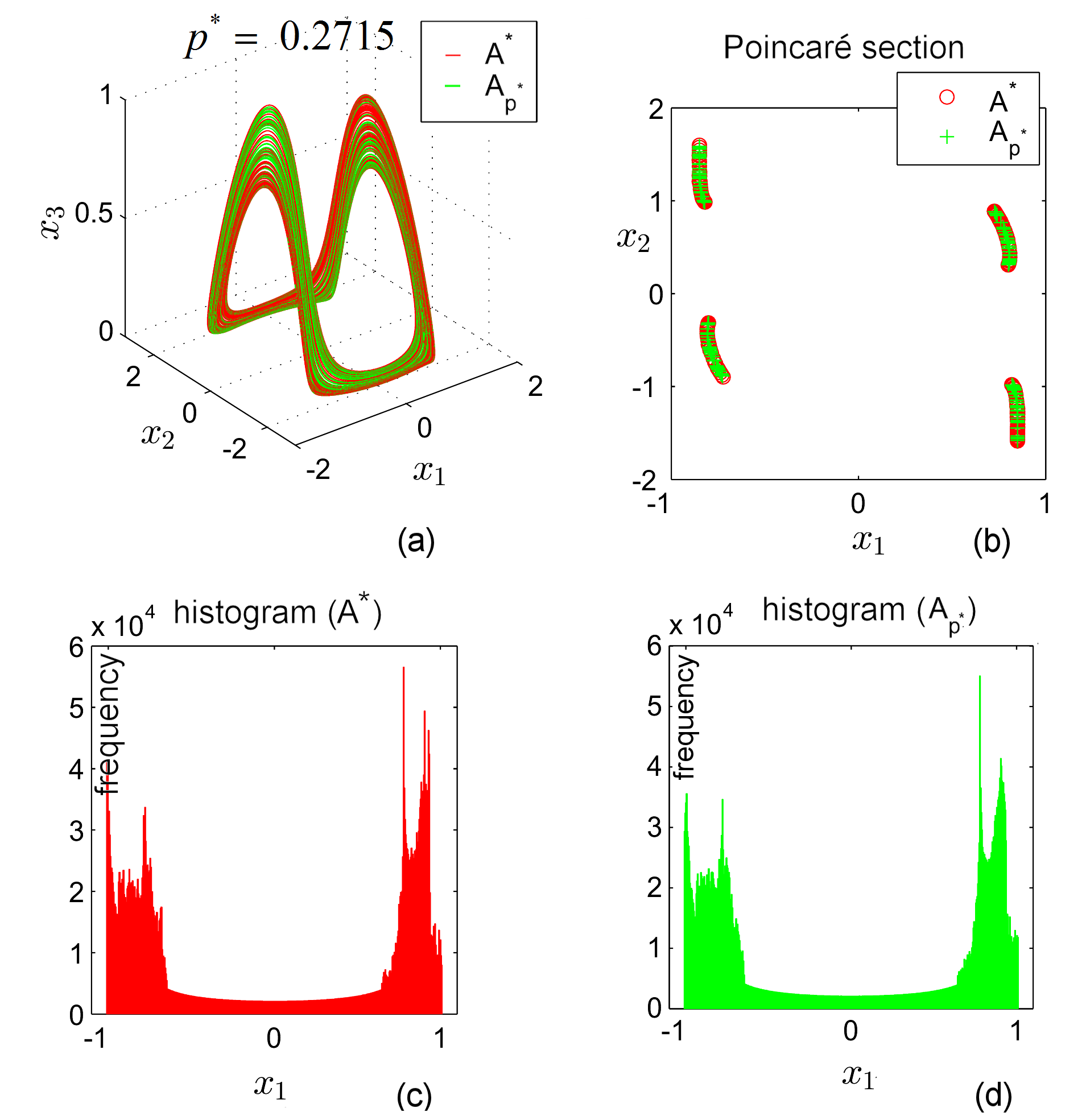}
\caption{Hidden chaotic attractor $H_3$ of the RF system \eqref{rf} corresponding to $p^*=0.2715$, obtained using the PS algorithm, with scheme $[1p_1,1p_2]$, $\mathcal{P}_2=\{0.265,0.278\}$. (a) Overplots of generated attractor $A^*$ (red plot) and averaged attractor $A_{p^*}$ (green plot). (b) Overplots of Poincar\'{e} sections with plane $x_3=0.3$, corresponding to $A^*$ and $A_{p^*}$. (c) Histogram with 512 bars of the first component $x_1$ of $A^*$ (red plot). (d) Histogram with 512 bars of the first component $x_1$ of $A_{p^*}$ (green plot).}
\label{fig10}
\end{center}
\end{figure}

\end{document}